\documentclass[11pt]{llncs}

\usepackage{amsmath}
\usepackage[mathscr]{eucal}
\usepackage{amssymb}
\usepackage{graphicx}
\usepackage{latexsym}

\usepackage{tikz}

\usepackage{color}
\usepackage{framed}

\ifx\suppress\undefined
\newcommand{\TODO}[1]{
\typeout{WARNING!!! there is still a TODO left}
\marginpar{\textbf{!TODO: }\emph{#1}}
}
\else
\newcommand{\TODO}[1]{}
\fi

\ifx\suppress\undefined
\newenvironment{todo}[1]{\noindent\rule{.3\textwidth}{1pt}\TODO{#1}\\}{\\\rule{.3\textwidth}{1pt}}
\else

\fi

\ifx\suppress\undefined
\newcommand{\NOTE}[1]{
\typeout{WARNING!!! there are still DRAFT NOTES left}
\marginpar{!DRAFT}\emph{\textbf{DRAFT NOTES:} #1}
}
\else
\newcommand{\NOTE}[1]{}
\fi

\newcommand{\rank}{\mathsf{rank}}

\title{On Weak Odd Domination and Graph-based~Quantum~Secret~Sharing}

\author{Sylvain Gravier\inst{1,2}\and J\'er\^ome Javelle\inst{3}\and Mehdi Mhalla\inst{1,3} \and Simon Perdrix\inst{1,3}}

\date{}

\institute{CNRS \and Institut Fourier, University of Grenoble, France \and LIG, University of Grenoble, France}

\begin{document}

\maketitle

\\
\begin{abstract}
A {weak odd dominated} (WOD) set in a  graph is a subset $B$ of vertices for which there exists a distinct set of vertices $C$ such that every vertex in $B$ has an odd number of neighbors in $C$. 
 We point out the connections of weak odd domination with odd domination, $[\sigma,\rho]$-domination, and perfect codes. We introduce bounds on $\kappa(G)$, the maximum size of WOD sets of a graph $G$, and on $\kappa'(G)$, the minimum size of non WOD sets of $G$. Moreover, we prove that the corresponding decision problems are NP-complete.

The study of weak odd domination is mainly motivated by the design of graph-based quantum secret sharing protocols
: a graph $G$ of order $n$ corresponds to a secret sharing protocol which threshold is $\kappa_Q(G) = \max(\kappa(G), n-\kappa'(G))$. These graph-based protocols are very promising in terms of physical implementation, however 
all such graph-based protocols studied in the literature have 
quasi-unanimity thresholds (i.e. $\kappa_Q(G)=n-o(n)$ where $n$ is the order of the graph $G$ underlying the protocol). In this paper, we show using probabilistic methods, the existence of graphs with smaller $\kappa_Q$ (i.e. $\kappa_Q(G)\le 0.811 n$ where $n$ is the order of $G$). We also prove that deciding for a given graph $G$ whether $\kappa_Q(G)\le k$ is NP-complete, which means that one cannot efficiently double check that a graph randomly generated has actually a $\kappa_Q$ smaller than $0.811n$.

\end{abstract}

\section{Introduction}
\label{intro}

\subsubsection*{Odd domination.} 
Odd domination is a variant of domination in which, given a graph $G=(V,E)$,  a set $C\subseteq V$ oddly dominates its closed odd neighborhood $Odd[C] := \bigtriangleup_{v\in C} N[v] = \{u\in V, |N[u] \cap C| = 1\mod2\}$ defined as the symmetric difference of the closed neighborhoods $N[v]= \{v\}\cup N(v)$ of the vertices $v$ in $C$, where $N(v) =\{u\in V, (u,v)\in E\}$ is the (open) neighborhood of $v$. 
An odd dominating set is a set of vertices $C\subseteq V$ such that $Odd[C]=V$. Odd dominating sets have been largely studied in the literature \cite{Amin19981,Caro} in particular for their role in the sigma-game \cite{Sutner,Sutner_sigma}. It has been noticeably proven that every graph contains at least one odd-dominating set \cite{Sutner} and that deciding whether a graph contains an odd dominating set of size at most $k$ is NP-complete \cite{Sutner}. 

Odd domination is a particular instance of the general framework of $[\sigma,\rho]$-domination \cite{Halldorsson:1999uq,Telle}. Given $\sigma, \rho \subseteq \mathbb N$, a $[\sigma, \rho]$-dominating set in a graph $G = (V,E)$ is a set $C\subseteq V$ such that $\forall v\in C, |N(v)\cap C|\in \sigma$, and $\forall v\in V\setminus C,  |N(v)\cap C|\in \rho$. Among others,   domination, independent set, perfect code, and odd domination problems can be formulated as $[\sigma, \rho]$-domination problems. In particular, odd domination corresponds to $[\textup{EVEN},\textup{ODD}]$-domination\footnote{Notice that odd domination is not a $[\textup{ODD}, \textup{ODD}]$-domination because open neighborhood are considered in the $[\sigma,\rho]$-domination instead of the closed neighborhood in the odd domination.}, where $\textup{EVEN} = \{2n, n\in \mathbb N\}$ and $\textup{ODD} = \mathbb N\setminus \textup{EVEN}$. The role of the parameters $\sigma$ and $\rho$ in the computational complexity of the corresponding decision problems  have been studied in the literature \cite{Telle}.

We consider a weaker version of odd domination which does not fall within the $[\sigma,\rho]$-domination framework. A weak odd dominated (WOD) set  is a set $B\subseteq V$ for which there exists $C\subseteq V\setminus B$ such that $B\subseteq Odd[C]$. Notice that, since $B\cap C=\emptyset$, $B\subseteq Odd[C]$ if and only if  $B\subseteq Odd(C) :=  \bigtriangleup_{v\in C} N(v) = \{u\in V, |N(u) \cap C| = 1\mod2\}$. Roughly speaking,  $B$ is a weak odd dominated set if it is oddly dominated by a set $C$ which does not intersect $B$. Weak odd domination does not fall within the $[\sigma, \rho]$-domination framework because, intuitively, a weak odd dominated set is not oddly dominated by its complementary set (as it would be in the $[\mathbb N, \textup{ODD}]$-domination) but by a subset of its complementary set. 

We consider two natural optimization problems related to weak odd dominated sets of a given graph $G$:  finding the size $\kappa(G)$ of the greatest WOD set and  finding  the size $\kappa'(G)$ of the smallest set which is not a WOD set. The greatest WOD set has a simple interpretation in a variant of the sigma-game: given a graph $G$, each vertex has three possible states: `on', `off', and  `broken'; when one plays on a vertex $v$, it makes the vertex $v$ `broken' and flips the states `on'/`off' of its neighbors. In the initial configuration all vertices are `off'. The size $\kappa(G)$ of the greatest WOD set corresponds to the greatest number of (unbroken) `on' vertices one can obtain.

  In section 2
  , we illustrate the weak odd domination by the computation of $\kappa$ and $\kappa'$ on a particular family of graphs. Moreover, we give non trivial bounds on these quantities, and  show that the corresponding decision problems are NP-complete.

\subsubsection*{Graph-based Quantum secret sharing.}

Our main motivation for studying weak odd dominated sets is not the variant of the sigma-game but  their crucial  role in graph-based protocols for quantum secret sharing. A quantum secret sharing  scheme \cite{gottesman} consists in sharing a quantum state among $n$ players such that authorized sets of players can reconstruct the secret. The protocol admits a threshold $k$ if any set of at least $k$ players can reconstruct the secret whereas any set of less than $k-1$ players has no information about the secret. Notice that a direct consequence of the no-cloning theorem \cite{WZ82} is that the threshold of any quantum secret sharing protocol on $n$ players must be greater than $n/2$. 

In a graph-based quantum secret sharing \cite{MS,JMP} the quantum state shared by the players is characterized by a simple undirected graph. In section \ref{QSS}, we show that the threshold of such a quantum secret sharing protocol based on a graph $G$ of order $n$ is $\kappa_Q(G)+1$, where $\kappa_Q(G) =  \max(\kappa(G), n-\kappa'(G))$.  Graph-based quantum secret sharing are very promising in terms of physical implementation \cite{Petal,Wetal}, however all the known graph-based secret sharing protocols have a threshold $n-o(n)$ (the best known protocol has a threshold $n-n^{0.68}$ \cite{JMP}), where $n$ is the number of players. On the other hand, it has been proved  that the threshold of any graph-based quantum secret sharing protocol on $n$ players is at least $\frac n2+\frac n{156}$  \cite{JMP}.

 In section \ref{sec:smallkappaQ}, we prove that there exists a family $\{G_i\}$ of graphs such $\kappa_Q(G_i)\le 0.811n_i$ where $n_i$ is the order of $G_i$. It crucially shows that graph-based quantum secret sharing protocols are not restricted to quasi-unanimity thresholds.    
We actually prove that that almost all the graphs have such a `small' $\kappa_Q$: if one picks  a random a graph $G$ of order $n$ (every edge occurs with probability $1/2$), then $\kappa_Q(G)\le 0.811n$ with probability greater than $1-\frac1n$.  We also prove that, given a graph $G$ and a parameter $k$,  deciding whether $\kappa_Q(G) \ge k$ is NP-complete. As a consequence, one cannot efficiently verify that a particular randomly generated graph has actually a `small' $\kappa_Q$.

\subsubsection*{Combinatorial Properties of Graph States.}

The development and the study of  graph-based protocols \cite{MS,KFMS,JMP,PS12,JMP11} have already pointed out  deep connections between graph theory and quantum information theory. For instance, it has been shown  \cite{informationflow} that a particular notion of flow \cite {DK06,BKMP,MP08,MMPST} in the underlying graph captures the flow, during the protocol, of the information contained in the secret from the dealer -- who encodes the secret and sends the shares -- to the authorized sets of players. The results presented in this paper contribute to these deep connections: we show that weak odd domination is a key concept for studying the properties of graph-based quantum secret sharing protocols.  

The study of graph-based protocols also contributes, as a by-product, to a better understanding of the combinatorial properties of a particular class of quantum states, called graph states \cite{HDERNB06-survey}. 
  The graph state formalism is a very powerful tool which is used in several areas of quantum information processing. Graph states provide a universal resource for quantum computing \cite{BRQC,VdNMDB06} and are also used in quantum correction codes \cite{SW03,ChuangShor} for instance.    
  Moreover, they are very promissing in terms of physical implementation \cite{Petal,Wetal}. As a consequence, progresses in the knowledge of the fundamental properties of graph states can potentially impact not only quantum secret sharing but  a wide area of quantum information processing.

\section{Weak Odd Domination}
\label{WOD}

We define one of the central notions in this paper: the weak odd domination. A set $B$ of vertices is a weak odd dominated (WOD) set if it is contained in the odd neighborhood of some set of vertices $C$ which does not intersect $B$: 
\begin{definition}
\label{WOD}
Given a simple undirected graph $G=(V,E)$, $B \subseteq V$ is a Weak Odd Dominated (WOD) set if there exists $C \subseteq V \setminus B$ such that $B \subseteq Odd(C)$, where $Odd(C) = \{v\in V, |N(v)\cap C| = 1 \mod2\}$.  
\end{definition}

Sets which are not WOD sets enjoy a noticeable  characterization: they contain  an odd set together with its odd neighborhood (proof is given in appendix):

\begin{lemma}
\label{accessing}
Given a graph $G=(V,E)$, $B$ is not a WOD set if and only if there exists $D\subseteq B$ such that $|D|=1\mod2$ and $Odd(D)\subseteq B$. 
\end{lemma}

This is clear from the definition that any subset of a WOD set is a WOD set and that any superset of a non WOD set is not a WOD set. As a consequence, we focus our attention on finding the greatest WOD set and the smallest non WOD set by considering the following quantities:

\begin{definition}
For a given graph $G$, let
$$
\kappa(G) = \max_{B~\text{WOD}} |B| \qquad\qquad
\qquad\kappa'(G) =  \min_{B~\text{not WOD}} |B|$$

\end{definition}

In the rest of this section, $\kappa$ and $\kappa'$ 
are computed for a particular  family of graphs, then we introduce bounds on these quantities in the general case, from which we prove the NP completeness of the decision problems associated with $\kappa$ and $\kappa'$.

To illustrate the concept of weak odd domination, we consider the following family of graphs: for any $p,q\in \mathbb N$, let $G_{p,q}$ be the complete $q$-partite graph where each independent set is of size $p$ (the proof is given in Appendix). 
$G_{p,q}$ is of order $n=pq$.

\begin{lemma}
\label{Gpq}
For any $p,q\in \mathbb N$, 

$\begin{array}{rclcrcll}
\kappa(G_{p,q})&=&n-p & \text{and} & \kappa'(G_{p,q})&=& q  & \text{ if $q=1\bmod 2$}\\
\kappa(G_{p,q})&=&\max(n-p,n-q)& \text{and} & \kappa'(G_{p,q})&=& p+q+1& \text{ if $q=0\bmod 2$}
\end{array}$
\label{gpq}
\end{lemma}

We show that the sum of $\kappa(G)$ and $\kappa'(\overline G)$ is always greater than the order of the graph $G$. The proof is based on the duality property  that the complement of a non-WOD set in $G$ is a WOD set in $\overline G$, the complement graph of $G$. 

\begin{lemma}\label{lem:dual}
Given a graph $G=(V,E)$, if $B\subseteq V$ is not a WOD set in $G$ then $V\setminus B$ is a WOD set in $\overline G$.  
\end{lemma}

\begin{proof} Let $B$ be a non-WOD set in $G$. $\exists D\subseteq B$ such that $|D|=1\bmod 2$ and $Odd_G(D)\subseteq B$. As a consequence, $\forall v\in V\setminus B$, $|N_G(v)\cap D|=0 \bmod 2$. Since $|D|=1\bmod 2$, $\forall v\in V\setminus B$, $|N_{\overline G}(v)\cap D|=1\bmod 2$. Thus, $V\setminus B$ is a WOD set in $\overline G$. 
\hfill $\Box$
\end{proof}

\begin{theorem}\label{thm:dual}
For any graph $G$ of order $n$, $\kappa'(G)+\kappa(\overline G)\ge n$.
\end{theorem}

\begin{proof}
There exists a non-WOD set $B\subseteq V$ such that $|B|= \kappa'(G)$. According to Lemma \ref{lem:dual}, $V\setminus B$ is WOD in $\overline G$, so $n-|B|\le \kappa(\overline G)$, so $n-\kappa'(G)\le \kappa(\overline G)$. 
\hfill $\Box$
\end{proof}

For any vertex $v$ of a graph $G$, its (open) neighborhood $N(v)$ is a WOD set, whereas, according to lemma \ref{accessing} its closed neighborhood (i.e. $N[v] = \{v\} \cup {N}(v)$) is a non-WOD set, as a consequence:

$$\kappa(G)\ge \Delta~~~~~~~~~~~~~~~~~\kappa'(G)\le \delta+1$$
where $\Delta$ (resp. $\delta$) denotes the maximal (resp. minimal) degree of the graph $G$. 

In the following, we prove an upper bound on $\kappa(G)$ and a lower bound on $\kappa'(G)$.

\begin{lemma}\label{lem:ubound}
For any graph $G$ of order $n$ and degree $\Delta$, $\kappa(G)\le \frac{n.\Delta}{\Delta +1}$.
\end{lemma}

\begin{proof}
Let $B\subseteq V$ be a WOD set, according to Definition \ref{WOD}, $\exists C\subseteq V\setminus B$ such that $B \subseteq Odd(C)$. $|C|\le n-|B|$ and $|B|\le |Odd(C)|\le \Delta.|C|$, so $|B|\le \Delta.(n-|B|)$. It comes that $|B|\le \frac{n.\Delta}{\Delta+1}$, so $\kappa(G)\le \frac{n.\Delta}{\Delta+1}$.
\hfill $\Box$
\end{proof}

In the following we prove that this bound is reached only for graphs having a perfect code. A graph $G=(V,E)$ has a perfect code if  there exists $ C\subseteq V$ such that $C$ is an independent set and every vertex in $V\setminus C$ has exactly one neighbor in $C$.  

\begin{theorem}
For any graph $G$ of order $n$ and degree $\Delta$, $\kappa(G)=\frac{n.\Delta}{\Delta +1}$ if and only if $G$ has a perfect code $C$ such that  $\forall v\in C$, $d(v)=\Delta$. 
\end{theorem}

\begin{proof}
($\Leftarrow$) Let $C$ be a perfect code of $G$ such that  $\forall v\in C$, $\delta(v)=\Delta$. $V\setminus C$ is a WOD set since $Odd(C) = V\setminus C$. Moreover $|V\setminus C|= \frac{n\Delta}{\Delta+1}$, so $\kappa(G)\ge \frac{n.\Delta}{\Delta +1}$. According to Lemma \ref{lem:ubound}, $\kappa(G)\le \frac{n\Delta}{\Delta+1}$, so $\kappa(G)= \frac{n\Delta}{\Delta+1}$.  \\
($\Rightarrow$) Let $B$ be a WOD set of size $\frac{n.\Delta}{\Delta +1}$. There exists  $C\subseteq V\setminus B$ such that $B\subseteq Odd(C)$. Notice that $|C|\le n-\frac{n.\Delta}{\Delta +1}=\frac{n}{\Delta +1}$. Moreover $|C|.\Delta\ge |Odd(C)| \ge  |B|$, so  $|C|=\frac{n}{\Delta +1}$. It comes that $|B|=|B\cap Odd(C)|\le \sum_{v\in C}d(v)\le \Delta.\frac{n}{\Delta+1}= |B|$. Notice that if $C$ is not a perfect code the first inequality is strict, and if $\exists v\in C$, $d(v)<\Delta$, the second inequality is strict. Consequently, $C$ is a perfect code and $\forall v\in C$, $d(v)=\Delta$.  
\hfill $\Box$
\end{proof}

\begin{corollary}\label{cor:kappa}
Given a $\Delta$-regular graph $G$, $\kappa(G) = \frac {n\Delta} {\Delta+1}$ if and only if  $G$ has a perfect code. 
\end{corollary}

We consider the problem {\bf MAX\_WOD} which consists in deciding, given a graph $G$ and an integer $k\ge 0$, whether $\kappa(G)\ge k$.

\begin{theorem}
{\bf MAX\_WOD} is NP-Complete.
\end{theorem}

\begin{proof}
${\bf MAX\_WOD}$ is in the class NP since a WOD set $B$ of size $k$ is a YES certificate. Indeed,  deciding whether the certificate $B$ is WOD or not can be done in polynomial time by  solving for $X$ the linear equation $\Gamma_{V\setminus B}. X = 1_B$ in $\mathbb F_2$, where $1_B$ is a column vector of dimension $|B|$ where all entries are $1$, and $\Gamma_{V\setminus B}$ is the cut matrix, i.e. a submatrix of the adjacency matrix of the graph which columns correspond to the vertices in $V\setminus B$ and rows to those in $B$. In fact, $X\subseteq V\setminus B$ satisfies $\Gamma_{V\setminus B}. X = 1_B$ if and only if ($X\subseteq V\setminus B$ and $B\subseteq Odd(X)$) if and only if $B$ is WOD. 
For the completeness, given a 3-regular graph, if $\kappa(G) \ge \frac 3 4 n$ then $\kappa(G) =  \frac 3 4 n$ (since $\kappa(G) \le  \frac{n\Delta}{\Delta+1}$ for any graph). Moreover, according to Corollary \ref{cor:kappa}, $\kappa(G)= \frac 3 4 n$ if and only if $G$ has a perfect code. Since the problem of deciding whether a $3$-regular graph has a perfect code  is known to be NP complete (see \cite{perf_codes} and \cite{KMP}), so is ${\bf MAX\_WOD}$.
\hfill $\Box$
\end{proof}

Now we introduce a lower bound on $\kappa'$.

\begin{lemma}
For any graph $G$,  $\kappa'(G)\ge \frac n{n-\delta}$ where $\delta$ is the minimal degree of $G$. 
\end{lemma}

\begin{proof} According to Theorem \ref{thm:dual}, $\kappa'(G)\ge n-\kappa(\overline G)$. Moreover, thanks to Lemma \ref{lem:ubound}, $n-\kappa(\overline G) \ge  n-\frac{n\Delta({\overline G})}{\Delta({\overline G})+1} =n- \frac{n(n-1-\delta(G))}{n-\delta(G)}= \frac{n}{n-\delta}$. 
\hfill $\Box$
\end{proof}

This bound is reached for the regular graphs for which their complement graph has a perfect code, more precisely:

\begin{theorem}\label{thm:kappap}
Given $G$ a $\delta$-regular graph such that $\frac {n}{n-\delta}$ is odd, $\kappa'({G}) = \frac n {n-\delta} $ if and only if  $\overline{G}$ has a perfect code. 
\end{theorem}

\begin{proof}
($\Leftarrow$) Let $C$ be a perfect code of $\overline G$. Since $|C|=\frac n{\Delta(\overline G)+1}=\frac n{n-\delta}=1~\bmod{2}$, $Odd_{{G}}(C)\subseteq C$, thus $C$ is a non-WOD set in ${G}$, so $ \kappa'({G})\le \frac n {n-\delta}$. Since $\kappa'(G) \ge \frac n {n-\delta}$ for any graph, $\kappa'(G) = \frac n {n-\delta}$\\
($\Rightarrow$) Let $B$ be a non-WOD set of size $\frac n{n-\delta}$ in ${G}$. $\exists D\subseteq B$ such that $|D|=1\bmod 2$ and $Odd_{{G}}(D)\subseteq B$. According to Lemma \ref{lem:dual},   $V\setminus B\subseteq Odd_{\overline G}(D)$, so $|Odd_{\overline G}(D)|\ge \Delta(\overline G)\frac{n }{n-\delta}$, which implies that  $|D|.\Delta(\overline G) \ge \Delta(\overline G)\frac{n}{n-\delta}$. As a consequence, $|D|= \frac n {n-\delta}$ and since every vertex of $V\setminus B$ (of size $\Delta(\overline G)\frac{n}{n-\delta}$) in $\overline G$ is connected to $D$, $D$ must be a perfect code. 
\hfill $\Box$
\end{proof}

We consider the problem ${\bf MIN\_\neg WOD}$ which consists in deciding, given a graph $G$ and an integer $k\ge 0$, whether $\kappa'(G)\le k$?

\begin{theorem}
 ${\bf MIN\_\neg WOD}$ is NP-Complete.
\end{theorem}

\begin{proof}
${\bf MIN\_\neg WOD}$ is in the class NP since a non-WOD set of size $k$ is a YES certificate.
For the completeness, given a 3-regular graph $G$,  if $\frac{n}{4}$ is odd then according to Theorem \ref{thm:kappap}, $G$ has a perfect code if and only if $\kappa'(\overline G)=\frac{n}{4}$. If   $\frac{n}{4}$ is even, we add a $K_4$ gadget to the graph $G$. Indeed, $G\cup K_4$ is a 3-regular graph and $\frac{n+4}{4}=\frac{n}{4}+1$ is odd. Moreover, $G$ has a perfect code if and only if $G\cup K_4$ has a perfect code if and only if $\kappa'(\overline {G\cup K_4})=\frac{n}{4}+1$. Since deciding whether a $3$-regular graph has a perfect code  is known to be NP complete, so is ${\bf MIN\_\neg WOD}$
\hfill $\Box$
\end{proof}

\section{From WOD sets to quantum secret sharing}
\label{QSS}

A quantum secret sharing  scheme \cite{gottesman} consists in: $(i)$ encoding a quantum state (the secret) into a $n$-partite quantum state and $(ii)$ sending to each of the $n$ players one part of this encoded quantum state 
 such that authorized sets of players can collectively reconstruct the secret. 
In \cite{MS}, Markham and Sanders introduced a particular family of quantum secret sharing protocols where the $n$-partite quantum state shared by the players is represented by a graph (such quantum states are called graph states \cite{BRQC}). They investigated the particular case where the secret is classical and they have shown that a set of players can perfectly recover a quantum secret in a protocol described by a graph $G$ if and only if they can recover a classical secret in both protocols described by $G$ and $\overline G$. In \cite{informationflow}, graphical conditions have been proven for a set of players to be able to recover a classical secret or not. Rephrased in terms of weak odd domination, they proved that a non-WOD set of players can recover a classical secret, whereas a WOD set cannot recover a classical secret. As a consequence, any set of more than $\kappa_Q(G) = \max ( \kappa(G), \kappa(\overline{G}))$ players can recover a quantum secret in the protocol described by $G$ since they can reconstruct the classical secret in both protocols $G$ and $\overline G$. Moreover, there exists a set of $B$ of players such that $|B|\le  \kappa_Q(G)$ which cannot recover the secret in $G$ or in $\overline G$, thus $B$ cannot perfectly recover  the quantum secret. As a consequence, $\kappa_Q(G)+1$ is nothing but the optimal threshold from which any set of more than $\kappa_Q(G)$ players can recover the quantum secret in the protocol described by $G$. To guarantee that $\kappa_Q(G)+1$ is actually the threshold of the protocol, i.e. any set of less than $\kappa_Q(G)$ players have no information about the secret, a preliminary encoding of the secret is performed (see \cite{JMP} for details).

In the following, we prove that deciding, given a graph $G$ and $k\ge0$, whether $\kappa_Q(G)\ge k$ is NP complete (Theorem \ref{NPkappaQ}). The proof consists in a reduction from the problem ${\bf MIN\_\neg WOD}$, and requires the following two ingredients: an alternative characterization of $\kappa_Q$ in terms of $\kappa$ and $\kappa'$ (Lemma \ref{lem:alter}); and the evaluation of $\kappa$ and $\kappa'$ for particular graphs consisting of multiple copies of a same graph (Lemma \ref{lem:copies}).

\begin{lemma}\label{lem:alter}
Given a graph $G$ of order $n$, $\kappa_Q(G) = \max ( \kappa(G), n-\kappa'(G) )$
\end{lemma}

\begin{proof}
Lemma \ref{lem:dual} gives $\kappa(\overline{G}) \geq n-\kappa'(G)$.
We show that if the value of $\kappa_Q(G)$ is not given by the value of $\kappa(G)$, we must have the equality between $\kappa(\overline{G})$ and $n-\kappa'(G)$.
In other terms, we want to show that $\kappa(G) < \kappa(\overline{G}) \Rightarrow \kappa(\overline{G}) = n-\kappa'(G)$.

We assume $\kappa(G) < \kappa(\overline{G})$.
There exists a set $B \subseteq V$ of size $\kappa(\overline{G})$ dominated by some $C \subseteq V \setminus B$.
We claim that $|C| = 1 \bmod 2$, otherwise $B$ would be WOD in $G$ and $\kappa(G) \geq \kappa(\overline{G})$.
Then the set $V\setminus B$ is non-WOD in $G$ since it contains $C \cup Odd(C)$ (see Lemma \ref{accessing}).
Consequently, $\kappa'(G) \leq |V\setminus B|$ which can be written $\kappa(\overline{G}) \leq n-\kappa'(G)$.
\hfill $\Box$
\end{proof}

\begin{lemma}\label{lem:copies}
For any graph $G$ and any $r>0$,
$\kappa(G^r) = r.\kappa(G)$ and 
$\kappa'(G^r) = \kappa'(G)$ 
 where $G^1=G$ and $G^{r+1}=G\cup G^r$.
\end{lemma}

\begin{proof}~
\\-- $[\kappa(G^r) = r.\kappa(G) \label{kappa_r}]$:
Let $B$ be a WOD set in $G$ of size $\kappa(G)$. $B$ is in the odd neighborhood of some $C \subseteq V$. Then the set $B_r \subseteq V(G^r)$ which is the union of sets $B$ in each copy of the graph $G$ is in the odd neighborhood of $C_r \subseteq V(G^r)$, the union of sets $C$ of each copy of $G$.
Therefore $B_r$ is WOD and $\kappa(G^r) \geq r.\kappa(G)$.
Now if we pick any set $B_0 \subseteq V(G^r)$ verifying $|B_0|> r.\kappa(G)$, there exists a copy of $G$ such that $|B_0 \cap G| > \kappa(G)$.
Therefore $B_0$ is a non-WOD set and $\kappa(G^r) \leq r.\kappa(G)$.
\\-- $[\kappa'(G^r) = \kappa'(G) \label{kappa_p_r}]$:
Let $B$ be a non-WOD set in $G$ of size $\kappa'(G)$. If we consider $B$ as a subset of $V(G^r)$ contained in one copy of the graph $G$, $B$ is a non-WOD set in $G^r$.
Therefore $\kappa'(G^r) \leq \kappa'(G)$.
If we pick any set $B \subseteq V(G^r)$ verifying $|B|< \kappa'(G)$, its intersection with each copy of $G$ verifies $|B \cap G| < \kappa'(G)$.
Thus, each such intersection is in the odd neighborhood of some $C_i$. So $B$ is in the odd neighborhood of $\bigcup_{i=1..r}C_i$. 
Consequently, $B_0$ is a WOD set in $G^r$ and $\kappa'(G^r) \geq \kappa'(G)$.
\hfill $\Box$
\end{proof}

We consider the problem {\bf QKAPPA} which consists in deciding, for a given graph $G$ and $k\ge 0$, whether $\kappa_Q(G)\ge k$, i.e. $\kappa(G)\ge  k$ or $\kappa'(G)\le  n-k$?

\begin{theorem}
\label{NPkappaQ}
{\bf QKAPPA} is NP-Complete.
\end{theorem}

\begin{proof}
{\bf QKAPPA} is in NP since a WOD set of size $k$ or a non-WOD set of size $n-k$ is a YES certificate.
For the completeness, we use a reduction to the problem ${\bf MIN\_\neg WOD}$.
Given a graph $G$ and any $k\ge0$, $\kappa_Q(G^{k+1})\ge (k+1)n-k \Leftrightarrow \Big(\kappa(G^{k+1})\ge (k+1)n-k \text{ or } \kappa'(G^{k+1})\le  k\Big) \Leftrightarrow \Big(\kappa(G)\ge n-1+\frac{1}{k+1} \text{ or } \kappa'(G)\ge k\Big) \Leftrightarrow \Big(\kappa(G) > n-1 \text{ or } \kappa'(G)\ge k\Big)$.
In the last disjunction, the first inequality $\kappa(G) > n-1$ is always false since for any graph $G$ of order $n$ we have $\kappa(G) \le n-1$.
Thus, the answer of the oracle call gives the truth of the second inequality $\kappa'(G) \ge k$ which corresponds to the problem ${\bf MIN\_\neg WOD}$.
As a consequence, {\bf QKAPPA} is NP-complete.
\hfill $\Box$
\end{proof}

\section{Graphs with small  $\kappa_Q$}
\label{sec:smallkappaQ}

\sloppy
Using lemma \ref{gpq},  the graphs $G_{\sqrt{n},\sqrt{n}}$ (when $n=p^2$) are such that $\kappa_Q(G_{\sqrt{n},\sqrt{n}})= n-\sqrt{n}$.

In this section, we prove using the asymmetric Lov\'asz Local Lemma \cite{lovasz} that there exists an infinite family of graphs $\{G_i\}$  such that $\kappa_Q(G_i)\le 0.811n_i$ where $n_i$ is the order of $G_i$. Moreover,  we prove that for a  random graph $G(n,1/2)$ (graph on $n$ vertices  where each pair of vertices have probability 1/2  to have an edge connecting them) with  $n\ge 100$,  
$\kappa_Q(G(n,1/2))\le 0.85n$ with  high probability. 

First we prove the following lemma:

\begin{lemma}
\label{problemme}
Given $k$ and $G=(V,E)$, if for any non empty set $D \subseteq V$, $|D\cup Odd(D)|> n-k$ and $|D\cup (V \setminus Odd(D))|> n-k$ then $\kappa_Q(G)< k$. 
\end{lemma}
\begin{proof}
Since $\forall D \subseteq V$ $|D\cup Odd(D)|> n-k$, $\kappa'(G)> n-k$. Let $B\subseteq V$, $|B|\ge k$, if $B$ is not WOD then $\exists C\subseteq V\setminus B$ such that  $B\subseteq Odd(C)$, so $(V \setminus Odd(C))\subseteq V\setminus B$ which implies $|C\cup (V \setminus Odd(C))|\le n-k$. 
\hfill $\Box$
\end{proof}

We use the asymmetric form of the Lov\'asz Local Lemma that can be stated as follows:
\begin{theorem}[Asymmetric Lov\'asz Local Lemma]
Let $\mathcal{A} = \{A_1, \cdots, A_n\}$ be a set of ÒbadÓ events in an arbitrary probability space and let $\Gamma(A)$ denote a subset of $\mathcal{A}$ such that $A$ is independent from all the events outside $A$ and $\Gamma(A)$.
If for all $A_i$ there exists $w(A_i) \in [0,1)$ such that
$Pr(A_i) \leq w(A_i) \prod_{B_j \in \Gamma(A_i)} (1-w(B_j))$
then we have
$Pr(\overline{A_1}, \cdots, \overline{A_n}) \geq \prod_{A_j \in \mathcal{A}} (1-w(A_j))$.
\end{theorem}

\begin{theorem}\label{thm:exist}

 There exists an infinite family of graphs $\{G_i\}$  such that $\kappa_Q(G_i)\le 0.811n_i$ where $n_i$ is the order of $G_i$.
\end{theorem}
\begin{proof}

Let $G(n,1/2)=(V,E)$ be a random graph. We  use the  asymmetric Lov\'asz local lemma to show that the probability that  for all non empty set $D \subseteq V$  $|D\cup Odd(D)|>(1-c)n$ and $|D\cup (V \setminus Odd(D))|> (1-c)n$ is positive for some constant $c$. This ensures by Lemma \ref{problemme} that $\kappa_Q(G)< cn$. 

We consider the ``bad" events $A_D \,:\, |Odd(D)\cup D| \le (1-c)n$ and  $A'_D \,:\, |Odd(D)\cup (V \setminus Odd(D))| \le (1-c)n$.
When $|D|>(1-c)n$, $Pr(A_D)=Pr(A'_D)=0$, therefore the previous events are defined for all $D$ such that $|D| \leq (1-c)n$.

For all $D$ such that $|D| \leq (1-c)n$, we want to get an upper bound on $Pr(A_D)$.
Let $|D|=dn$ for some $d \in \left( 0,1-c \right] $.
For all $u \in V \setminus D$, $Pr(``u \in Odd(D)") = \frac{1}{2}$.
If $D$ is fixed, the events $``u \in Odd(D)"$ when $u$ is outside $D$ are independent.
Therefore, if the event $A_D$ is true, any but at most $(1-c-d)n$ vertices outside $D$ are contained in $Odd(D)$.
There are $(1-d)n$ vertices outside $D$, then $Pr(A_D) = \left( \frac{1}{2} \right)^{(1-d)n} \sum_{k=0}^{(1-c-d)n} {(1-d)n \choose k} \leq \left( \frac{1}{2} \right)^{(1-d)n} 2^{(1-d)nH\left( \frac{1-c-d}{1-d} \right)} = 2^{(1-d)n \left[ H\left( \frac{c}{1-d} \right) - 1 \right]}$ where $H : t \mapsto -t\log_2(t)-(1-t)\log_2(1-t)$ is the binary entropy function.

Similarily, $Pr(A'_D) \leq 2^{(1-d)n \left[ H\left( \frac{c}{1-d} \right) - 1 \right]}$.

We consider that all the events can be dependent. 
For any $D\subseteq V$ such that $0<|D| \leq (1-c)n$, we define $w(A_D)=w(A'_D)=\frac{1}{r{n\choose |D|}}$.
 First, we verify that $Pr(A_D) \leq w(A_D) \prod_{D' \subseteq V, |D| \leq (1-c)n} (1 - w(A_{D'}))(1 - w(A'_{D'}))$.
 The product of the right-hand side of the previous equation can be written $p = \prod_{|D'| = 1}^{(1-c)n}  \left( 1 - \frac{1}{r{n  \choose |D'|}} \right)^{2{n  \choose |D'|}} = \left[ \prod_{|D'| = 1}^{(1-c)n}  \left( 1 - \frac{1}{r{n  \choose |D'|}} \right)^{r{n  \choose |D'|}} \right]^{\frac{2}{r}}$.
 The function $f : x \mapsto \left( 1-\frac{1}{x} \right)^x$ verifies $f(x) \geq \frac{1}{4}$ when $x \geq 2$, therefore $p \geq \left( \frac{1}{4} \right)^{\frac{2}{r}(1-c)n} = 2^{-\frac{4(1-c)n }{r}}$ for any $r \geq 2$.
 Thus, it is sufficient to have $2^{(1-d)n \left[ H\left( \frac{c}{1-d} \right) - 1 \right]} \leq \frac{1}{r{n  \choose dn }} 2^{-\frac{4(1-c)n }{r}}$.
 Rewriting this inequality gives $r{n  \choose dn } 2^{(1-d)n \left[ H \left( \frac{c}{1-d} \right) -1 \right] + \frac{4(1-c)n}{r}} \leq 1$.
 Thanks to the bound ${n  \choose dn } \leq 2^{nH\left( \frac{dn }{n } \right)}$ and after applying the logarithm function and dividing by $n$, it is sufficient that $(1-d) \left[ H \left( \frac{c}{1-d} \right) -1 \right] + H(d) + \frac{4(1-c)}{r} + \frac{\log_2{r}}{n} \leq 0$.
 If we take $r=n$, the condition becomes asymptotically $(1-d) \left[ H \left( \frac{c}{1-d} \right) -1 \right] + H(d) \leq 0$.
 
 Numerical analysis shows that this condition is true for any $c>0.811$ and for all $d \in \left( 0,1-c \right]$.
 Thus, thanks to the Lov\'asz Local Lemma, for any $c>0.811$, $Pr(\kappa_Q(G)< cn) \geq p \geq \left( \frac{1}{4} \right)^{\frac{2}{r}(1-c)n} > 0$, therefore there exists an infinite family of graphs $\{G_i\}$  such that $\kappa_Q(G_i)\leq 0.811n_i$ where $n_i$ is the order of $G_i$ for $n_i \geq N_0$ for some $N_0 \in \mathbb{N}$.
\hfill $\Box$
\end{proof}

Very recently \cite{PS12}, Sarvepalli proved that quantum secret sharing protocols based on graph states are equivalent to quantum codes. 
 Combining this result with the  Gilbert Varshamov bounds on quantum stabilizer codes \cite{FM04}, we can provide an alternative proof of theorem \ref{thm:exist}. However, we believe the use of the Lov\'asz Local Lemma offers several advantages: the proof is a purely graphical proof with a potential extension to the construction of good quantum secret sharing schemes using the recent development in the algorithmic version \cite{MT10} of the Lov\'asz Local Lemma. Moreover, the use of the probabilistic methods already offers a way of generating good quantum secret sharing protocols with high probability  by adjusting the parameters of the Lov\'asz Local Lemma:

 \begin{theorem} There exists $n_0$ such that for any $n>n_0$, a random graph $G(n,\frac12)$ has a $\kappa_Q$ smaller than $0.811n$ with high probability: $$Pr\left(\kappa_Q(G(n,\frac12))<0.811n\right)\ge 1-\frac1{n}$$
 
 \end{theorem}
 
 \begin{proof} The proof of the theorem is done as in the proof of theorem \ref{thm:exist}, by taking $c=0.811$ and $r=4\ln(2)(1-c)n^2$. It guarantees that for $n\ge 26681$, $(1-d) \left[ H \left( \frac{c}{1-d} \right) -1 \right] + H(d) + \frac{4(1-c)}{r} + \frac{\log_2{r}}{n} \leq 0$. Thus, for any $D\subseteq V$ such that $0<|D| \leq (1-c)n$, $Pr(A_D) \leq w(A_D) \prod_{D' \subseteq V, |D| \leq (1-c)n} (1 - w(A_{D'}))(1 - w(A'_{D'}))$.  
 Moreover the probability that none of the bad events occur is $Pr\left(\kappa_Q(G(n,\frac12))<0.811n\right)\ge \left( \frac{1}{4} \right)^{\frac{2}{r}(1-c)n} = \left( \frac{1}{4} \right)^{\frac{1}{2n\ln(2)}}=e^{-\frac1n}\ge 1-\frac1n$. 
\hfill $\Box$
 \end{proof}

 Notice that, as we proved that {\bf QKAPPA } is NP-Complete, one cannot double check easily whether a random graph generated by this method has actually a small $\kappa_Q$.

\section{Conclusion}

In this paper, we have studied the quantities $\kappa$, $\kappa'$ and $\kappa_Q$ that can be computed on graphs.
They correspond to the extremal cardinalities WOD and non-WOD sets can reach.
These quantities present strong connections with quantum information theory and the graph state formalism, and especially in the field of quantum secret sharing.

Thus, we have studied and computed these quantities on some specific families of graphs, and we deduced they are candidates for good  quantum secret sharing protocols. Then we have proven the NP-completeness of the decision problems associated with $\kappa$, $\kappa'$ and $\kappa_Q$. Finally we have proven the existence of an infinite family of graphs $\{G_i\}$  such that $\kappa_Q(G_i)\le 0.811n_i$ where $n_i$ is the order of $G_i$, and even that if one picks uniformly at random a graph $G$ of order $n$, then 
 $\kappa_Q(G)\le 0.811n$ with high probability. 
  An interesting question is to find an explicit family of graphs   $\{G_i\}$  such that $\kappa_Q(G_i)\le c n_i$ where $n_i$ is the order of $G_i$ and $c$ a constant smaller than 1.

A related question is also still open: is the problem of deciding whether the minimal degree up to local complementation is greater than $k$ NP-complete?
This problem seems very close to finding $\kappa'$ since it consists in finding the smallest set of vertices of the form $D \cup Odd(D)$ with $D \neq \emptyset$, without the constraint of parity $|D|=1 \bmod 2$ as for $\kappa'$.

\bibliographystyle{plain}

\begin{thebibliography}{10}

\bibitem{Amin19981}
Ashok~T. Amin, Lane~H. Clark, and Peter~J. Slater.
\newblock Parity dimension for graphs.
\newblock {\em Discrete Mathematics}, 187(1-3):1 -- 17, 1998.

\bibitem{BKMP} Daniel E. Browne, Elham KasheÞ, Mehdi Mhalla, and Simon Perdrix. Generalized ßow and determinism in measurement-based quantum computation. {\em New Journal of Physics (NJP)}, 9(8), 2007.


\bibitem{ChuangShor}
Salman Beigi, Isaac Chuang, Markus Grassl, Peter Shor, and Bei Zeng.
\newblock Graph concatenation for quantum codes.
\newblock {\em Journal of Mathematical Physics}, 52(2)(022201), 2011.


\bibitem{Caro}
Y.~Caro and W.~Klostermeyer.
\newblock The odd domination number of the odd domination number of a graph.
\newblock In {\em J. Comb. Math Comb. Comput.}, volume~44, pages 65--84, 2003.


\bibitem{DK06} Vincent Danos and Elham KasheÞ. Determinism in the one-way model.
{\em Physical Review A}, 74(052310), 2006.


\bibitem{gottesman}
Daniel Gottesman.
\newblock Theory of quantum secret sharing.
\newblock {\em Phys. Rev. A}, 61:042311, 2000.

\bibitem{FM04}
Keqin Feng and   Zhi Ma. A finite Gilbert-Varshamov bound for pure stabilizer quantum codes. {\em IEEE Transactions on Information Theory}, 50, pp 3323 - 3325, 2004. 

\bibitem{Halldorsson:1999uq}
Magn{\'u}s~M. Halld{\'o}rsson, Jan Kratochv{\'\i}l, and Jan~Arne Telle.
\newblock Mod-2 independence and domination in graphs, {\em International Journal of Foundations of Computer Science}, (3), vol 11, pp 355-363, 2000.

\bibitem{HDERNB06-survey}
Marc Hein, Wolfgang D\"ur, Jens Eisert, Robert Raussendorf, Maarten~Van den
  Nest, and Hans~J. Briegel.
\newblock Entanglement in graph states and its applications.
\newblock In {\em Proceedings of the International School of Physics ``Enrico
  Fermi'' on ``Quantum Computers, Algorithms and Chaos''}, 2005

\bibitem{JMP11}
J\'er\^ome  Javelle, Mehdi Mhalla, and Simon Perdrix. Classical versus Quantum Graph-based Secret Sharing.
{\em eprint:arXiv:1109.4731}, 2011.

\bibitem{JMP}
J\'er\^ome  Javelle, Mehdi Mhalla, and Simon Perdrix. New protocols and
lower bound for quantum secret sharing with graph states. In {\em Theory of Quantum Computation, Communication and Cryptography.
(TQCÕ12)},  To appear in Lecture Notes in Computer Science, 2012. 
{\em eprint:arXiv:1109.1487}, 2012.



\bibitem{informationflow}
Elham Kashefi, Damian Markham, Mehdi Mhalla, and Simon Perdrix.
\newblock Information flow in secret sharing protocols.
\newblock {\em EPTCS 9, 2009, pp. 87-97}, 09 2009.


\bibitem{KFMS} Adrian Keet,  Ben Fortescue,  Damian Markham,  and Barry C. Sanders. Quantum secret sharing with qudit graph states. {\em Phys. Rev. A, 82:062315}, 2010. 

\bibitem{KMP}
Sandi Klavzar, Uros Milutinovic, and Ciril Petr.
\newblock 1-perfect codes in sierpinski graphs.
\newblock {\em Bulletin of the Australian Mathematical Society}, 66:369--384,
  2002.



\bibitem{perf_codes}
Jan Kratochvil.
\newblock Perfect codes in general graphs.
\newblock {\em 7th Hungarian colloqium on combinatorics, Eger}, 1987.

\bibitem{lovasz}
L\'aszl\'o Lov\'asz.
\newblock Problems and results on 3-chromatic hypergraphs and some related
  questions.
\newblock In {\em Colloquia Mathematica Societatis Janos Bolyai}, pages
  609--627, 1975.

\bibitem{MS}
Damian Markham and Barry~C. Sanders.
\newblock Graph states for quantum secret sharing.
\newblock {\em Physical Review A}, 78:042309, 2008.

\bibitem{MMPST}Mehdi Mhalla, Mio Murao, Simon Perdrix, Masato Someya, and Peter
Turner. Which graph states are useful for quantum information processing?
In {\em Theory of Quantum Computation, Communication and Cryptography.
(TQCÕ11)},  To appear in Lecture Notes in Computer Science, 2011.

\bibitem{MP08} Mehdi Mhalla and Simon Perdrix. Finding optimal ßows e?ciently. {\em In the
35th International Colloquium on Automata, Languages and Programming
(ICALP)}, LNCS, volume 5125, pages 857Ð868, 2008.


\bibitem{MT10}Robin A. Moser, G‡bor Tardos. {\em A constructive proof of the general Lovasz Local Lemma}. Journal of the ACM (JACM), v.57 n.2, p.1-15,  2010. 



\bibitem{Petal}
Robert Prevedel, Philip Walther, Felix Tiefenbacher, Pascal Bohi, Rainer
  Kaltenbaek, Thomas Jennewein, and Anton Zeilinger.
\newblock High-speed linear optics quantum computing using active feed-forward.
\newblock {\em Nature}, 445(7123):65--69, 2007


\bibitem{BRQC}
Robert Raussendorf and Hans Briegel.
\newblock A one-way quantum computer.
\newblock {\em Physical Review Letters}, 86(22):5188--5191, 2001.


\bibitem{PS12}
Pradeep Sarvepalli. 
Non-Threshold Quantum Secret Sharing Schemes in the Graph State Formalism
{\em eprint:arXiv:1202.3433}, 2012. 

\bibitem{SW03}Dirk Schlingemann and Reinhard F. Werner
Quantum error-correcting codes associated with graphs.
{\em Phys. Rev. A 65, 012308},  2001. 


\bibitem{Sutner_sigma}
Klaus~Sutner.
\newblock The $\sigma$-game and cellular automata.
\newblock In {\em Amer. Math. Monthly}, pages 24--34, 1990.

\bibitem{Sutner}
Klaus Sutner.
\newblock Linear cellular automata and the garden-of-eden.
\newblock {\em The Mathematical Intelligencer}, 11:49--53, 1989.

\bibitem{Telle}
Jana .A. Telle.
\newblock Complexity of domination-type problems in graphs.
\newblock In {\em Nordic Journal of Computing}, volume~1, pages 157--171, 1994.


\bibitem{VdNMDB06}
Maarten Van~den Nest, Akimasa Miyake, Wolfgang D\"ur, and Hans~J. Briegel.
\newblock Universal resources for measurement-based quantum computation.
\newblock {\em Phys. Rev. Lett.}, 97:150504, 2006.


\bibitem{Wetal}
Philip Walther, Kevin~J. Resch, Terry Rudolph, Emmanuel Schenck, Harald
  Weinfurter, Vlatko Vedral, Markus Aspelmeyer, and Anton Zeilinger.
\newblock Experimental one-way quantum computing.
\newblock {\em Nature}, 434(7030):169--176, 2005

\bibitem{WZ82}
William K. Wootters and Wojciech H. Zurek, A single quantum cannot be cloned, \newblock  {\em Nature} 299, 802-803, 1982.

\end{thebibliography}

\appendix
\section{Appendix}

Proof of Lemma \ref{accessing}:\\\\
{\bf Lemma \ref{accessing}.}
\emph{Given a graph $G=(V,E)$, $B$ is not a WOD set if and only if there exists $D\subseteq B$ such that $|D|=1\mod2$ and $Odd(D)\subseteq B$. }

\begin{proof}
We express this lemma in the following way:\\
\emph{
Given a graph $G=(V,E)$, for any $B\subseteq V$, $B$ satisfies exactly one of the following properties:\\
$~~~i.$$~~\exists D\subseteq B, D\cup Odd(D)\subseteq B \text{~and~} |D|=1 \bmod{2}$\\
$~~~ii.$ $\exists C\subseteq V\setminus B, Odd(C)\cap B = B$\\}
For a given $B\subseteq V$, let $\Gamma_{B}$ be the cut matrix induced by $B$, i.e. the sub-matrix of the adjacency matrix $\Gamma$ of $G$ such that  the columns of $\Gamma_{B}$ correspond to the vertices in $B$ and its rows to the vertices in $V\setminus B$. $\Gamma_{B}$ is the matrix representation of the linear function which maps every $X\subseteq B$ to $\Gamma_{B}.X = Odd(X)\cap (V\setminus B)$, where the set $X$ is identified with its characteristic column vector.  Similarly, $\forall Y\subseteq V\setminus B$, $\Gamma_{V\setminus B}.Y = Odd(Y)\cap B$ where $\Gamma_{V\setminus B} = \Gamma_{B}^T$ since $\Gamma$ is symmetric. Moreover, notice that for any set $X,Y\subseteq V$, $|X\cap Y| \bmod{2}$ is given by the matrix product $Y^T.X$ where again sets are identified with their column vector representation. Equation $(i)$ is satisfied if and only if $\exists D$ such that $\left(\frac{B^T}{\Gamma_{B}}\right).D = \left(\frac 1 0\right)$ which is equivalent  to 
$\rank\left(\frac{B^T}{\Gamma_{B}}\right)  
= \rank\left(\frac{B^T~|~1}{\hspace{0.04cm}\Gamma_{B}~|~0}\right) 
= \rank\left(\frac{~0~~|~1}{\Gamma_{B}~|~0}\right) 
=\rank(\Gamma_{B})+1$. 
Thus $(i)$ is true iff $\pi(B) = 1$ where $\pi(B):=\rank\left(\frac{B^T}{\Gamma_{B}}\right) - \rank(\Gamma_{B})$. Similarly equation $(ii)$ is satisfied if and only if $\exists C$ such that $\Gamma_{V\setminus B}.C = B$ if and only if $\rank(\Gamma_{V\setminus B}|B) = \rank(\Gamma_{V\setminus B})$. Thus $(ii)$ is true if and only if $\pi(B)=0$. Since for any $B\subseteq V$,  $\pi(B)\in \{0,1\}$ it comes that either $(i)$ is true or $(ii)$ is true.
\hfill $\Box$
\end{proof}

Proof of the Lemma \ref{Gpq}:\\\\
{\bf Lemma \ref{Gpq}.}

\emph{For any $p,q\in \mathbb N$}, 

$\begin{array}{rclcrcll}
\kappa(G_{p,q})&=&n-p & \text{and} & \kappa'(G_{p,q})&=& q  & \text{ if $q=1\bmod 2$}\\
\kappa(G_{p,q})&=&\max(n-p,n-q)& \text{and} & \kappa'(G_{p,q})&=& p+q+1& \text{ if $q=0\bmod 2$}
\end{array}$
\label{gpq}

\begin{proof}
"f $q=1 \bmod 2$ 
\\-- $[\kappa(G_{p,q})\ge n-p]$: The subset $B$ composed of all the vertices but a maximal independent set (MIS) -- i.e. an independent set of size $p$ -- is in the odd neighborhood of each vertex in $V\setminus B$. 
Therefore $B$ is WOD and $|B| = n-p$.
Consequently, according to the previous definition, $\kappa(G_{p,q})\ge n-p$.
\\-- $[\kappa(G_{p,q})\le n-p]$: Any set $B$ such that $|B|>n-p$ contains at least one vertex from each of the $q$ MIS, i.e. a clique of size $q$.
Let $D\subseteq B$ be such a clique of size $|D|=q = 1\bmod{2}$.
Every vertex $v$ outside $D$ is connected to all the elements of $D$ but the one in the same MIS as $v$.
Thus $Odd(D)=\emptyset$. 
As a consequence, $B$ is non-WOD.
\\-- $[\kappa'(G_{p,q})\le q]$: $B$ composed of one vertex from each MIS is a non-WOD set (see previous item).
\\-- $[\kappa'(G_{p,q})\ge q]$: If $|B|<q$ then $B$ does not intersect all the MIS of size $p$, so $B$ is in the odd neighborhood of each vertex of such a MIS. So according to Definition \ref{WOD}, $B$ is WOD. 

"f $q=0 \bmod 2$ 

\begin{itemize}
\item $[\kappa(G)\ge max(n-p,n-q)]$: For $\kappa(G)\ge n-p$, see previous lemma. The subset $B$ composed of all the vertices but a clique of size $q$ (one vertex from each MIS) is in the odd neighborhood of $V\setminus B$. Indeed each vertex of $B$ is connected to $q-1=1\bmod{2}$ vertices of $V\setminus B$. So, according to Definition  \ref{WOD}, $B$ of size $n-q$ is WOD, as a consequence $\kappa(G)\ge n-q$.

\item $[\kappa(G)\le max(n-p,n-q)]$: Any set $B$ such that $|B|> max(n-p,n-q)$ contains at least one vertex from each MIS and moreover it contains a MIS $S$ of size $q$. Let $D\subseteq B\setminus S$ be a clique of size $q-1=1\bmod{2}$. Every vertex $u$ in $V\setminus B$ is connected to all the vertices in $D$ but  one, so $Odd(D)\subseteq B$.

\item  $[\kappa'(G)\le p+q-1]$: Let $S$ be an MIS. Let $B$ be the union of $S$ and  of a clique of size $q$. Let $D= B\setminus S$. $|D|=q-1=1\bmod{2}$. Every vertex  $u$ in $V\setminus B$ is connected to all the vertices of $D$ but one, so $Odd(D)\subseteq B$.

\item  $[\kappa'(G)\ge p+q-1]$: Let $|B|<p+q-1$. If  $B$ does not intersect all the MIS of size $p$, then $B$ is in the odd neighborhood  of each vertex of such a non intersecting MIS. If $B$ intersects all the MIS then it does not contain any MIS, thus there exists a clique $C\subseteq V\setminus B$ of size $q$. Every vertex in $B$ is in the odd neighborhood of $C$. 
\end{itemize}
\hfill $\Box$
\end{proof}

\end{document}